\documentclass[a4paper,11pt,draft]{article}
\usepackage{authblk}

\usepackage{lmodern}
\usepackage[T1]{fontenc}
\usepackage[utf8]{inputenc}
\usepackage{amsmath,amssymb,amsthm}
\usepackage{mathtools}
\usepackage[english]{babel}
\usepackage{csquotes}
\usepackage{tikz}
\usepackage{hyperref}
\usepackage[margin=1in]{geometry}
\usepackage[capitalize]{cleveref}

\hypersetup{
  colorlinks=true,
  draft=false
}

\newtheorem{theorem}{Theorem}
\newtheorem{lemma}[theorem]{Lemma}
\newtheorem{corollary}[theorem]{Corollary}
\newtheorem{claim}[theorem]{Claim}
\newtheorem{proposition}[theorem]{Proposition}
\theoremstyle{definition}
\newtheorem{definition}[theorem]{Definition}

\theoremstyle{remark}

\newenvironment{claimproof}[1][\proofname]{\begin{proof}[#1]}{\end{proof}}

\DeclareMathOperator{\dist}{dist}
\DeclareMathOperator{\supp}{supp}
\DeclareMathOperator{\dom}{dom}
\DeclareMathOperator{\Span}{span}

\DeclareMathOperator{\Tr}{Tr}

\DeclareMathOperator{\pre}{Pre}
\DeclareMathOperator{\post}{Post}
\DeclareMathOperator{\bias}{bias}
\DeclareMathOperator{\E}{\mathbb{E}}

\DeclareMathOperator{\ac0}{\mathsf{AC}^0}
\DeclareMathOperator{\dnf}{\mathsf{DNF}}
\DeclareMathOperator{\Mod}{\mathsf{Mod}}
\DeclareMathOperator{\reslin}{\mathsf{Res}[\oplus]}

\newcommand{\dual}[1]{{#1}^{*}}
\newcommand{\orth}[1]{{#1}^{\perp}}
\newcommand{\card}[1]{{\left|#1\right|}}

\newcommand{\iso}[1]{\mathbf{#1}}

\newcommand{\Q}{{\{0,1\}}}
\newcommand{\bbF}{\mathbb{F}}
\newcommand{\calP}{\mathcal{P}}
\newcommand{\restr}[2]{\left.#1\right|_{#2}}

\newcommand{\sreadonce}{strongly read-once}
\newcommand{\Sreadonce}{Strongly read-once}

\newcommand{\wreadonce}{weakly read-once}
\newcommand{\Wreadonce}{Weakly read-once}

\newcommand{\readonce}{read-once}
\newcommand{\Readonce}{Read-once}

\title{Linear Branching Programs and Directional Affine Extractors}

\author[1]{Svyatoslav Gryaznov \footnote{{\tt svyatoslav.i.gryaznov@gmail.com}. Supported by GA{\v C}R grant 19-05497S.}}
\author[2]{Pavel Pudl\'{a}k\footnote{{\tt pudlak@math.cas.cz}. Supported by GA{\v C}R grant 19-27871X.}}
\author[2]{Navid Talebanfard\footnote{{\tt talebanfard@math.cas.cz}. Supported by GA{\v C}R grant 19-27871X.}}

\affil[1]{St. Petersburg Department of V. A. Steklov Institute of Mathematics of the Russian Academy of Sciences}
\affil[2]{Institute of Mathematics, Czech Academy of Sciences}

\date{}

\begin{document}
\maketitle

\begin{abstract}

A natural model of \emph{\readonce{} linear branching programs} is a branching program where queries are $\bbF_2$ linear forms, and along each path, the queries are linearly independent. We consider two restrictions of this model, which we call \emph{weakly} and \emph{strongly} read-once, both generalizing standard read-once branching programs and parity decision trees. Our main results are as follows.

\begin{itemize}
    \item {\bf Average-case complexity.} We define a pseudo-random class of functions which we call \emph{directional affine extractors}, and show that these functions are hard on average for the \sreadonce{} model. We then present an explicit construction of such function with good parameters. This strengthens the result of Cohen and Shinkar (ITCS'16) who gave such average-case hardness for parity decision trees. Directional affine extractors are stronger than the more familiar class of affine extractors. Given the significance of these functions, we expect that our new class of functions might be of independent interest.

    \item {\bf Proof complexity.} We also consider the proof system $\reslin$ which is an extension of resolution with linear queries. A refutation of a CNF in this proof system naturally defines a linear branching program solving the corresponding search problem. Conversely, we show that a \wreadonce{} linear BP solving the search problem can be converted to a $\reslin$ refutation with constant blow up.
\end{itemize}

\end{abstract}

%\newpage
\section{Introduction}

Circuit complexity and proof complexity are two major lines of inquiry in complexity theory (see \cite{DBLP:books/daglib/0028687, kraj_2019} for extensive introductions). The former theme attempts to identify explicit Boolean functions which are not computable by small circuits from a certain restricted class, and the latter aims to find tautologies which are not provable by short proofs in a given restricted proof system. These seemingly unrelated topics are bound together in at least two different ways: via feasible interpolation where a circuit lower bound for a concrete computational problem implies proof size lower bounds (see, e.g., \cite{DBLP:conf/focs/HrubesP17}), and more fundamentally many proof systems have an underlying circuit class where proof lines come from. Notable examples are Frege, bounded depth Frege, and extended Frege systems where proof lines are De Morgan formulas, $\ac0$ circuits, and general Boolean circuits, respectively. Intuitively we expect that understanding a circuit class in terms of lower bounds and techniques should yield results in the proof complexity counterpart. This intuition has been supported by bounded depth Frege lower bounds using specialized Switching Lemmas (see, e.g., \cite{DBLP:journals/jacm/Hastad21}), the essential ingredient of $\ac0$ lower bounds.% (see, e.g., \cite{DBLP:conf/stoc/Hastad86,DBLP:journals/siamcomp/Hastad14a}).

\bigskip
\noindent{\bf $\ac0[2]$ circuits and $\reslin$ proof system.} It is not clear if this intuition should always hold. Lower bounds for $\ac0[p]$ circuits ($\ac0$ circuits with $\Mod_p$ gates) have been known for a long time \cite{MR897705,DBLP:conf/stoc/Smolensky87} yet lower bounds for bounded depth Frege systems with modular gates still elude us. Perhaps this failure is not too surprising since our understanding of $\ac0[p]$ circuits is not of the same status as our understanding of $\ac0$. For example, even for $\ac0[2]$, that is $\ac0$ with parity gates, no strong average-case lower bound is known. Settling such bounds is an important challenge, since Shaltiel and Viola \cite{DBLP:journals/siamcomp/ShaltielV10} showed that for standard worst-case to average-case hardness amplification techniques to work, the circuit class is required to compute the majority function, which is not the case for $\ac0[2]$. Several works have highlighted the special case of $\ac0\circ \Mod_2$, where the parity gates are next to the input \cite{DBLP:journals/eccc/ServedioV12,DBLP:conf/innovations/AkaviaBGKR14,DBLP:conf/innovations/CohenS16}. Among these works we pay special attention to the result of Cohen and Shinkar \cite{DBLP:conf/innovations/CohenS16} who considered the depth-3 case of this problem and proved a strong average-case hardness for the special case of parity decision trees. The more general case of $\dnf \circ \Mod_2$ remains open.

In the proof complexity parallel, a special case of $\ac0[2]$-Frege was suggested by Itsykson and Sokolov \cite{DBLP:journals/apal/ItsyksonS20}. They considered the system $\reslin$ that is an extension of resolution which reasons about disjunctions of linear equations over $\bbF_2$, which we call linear clauses. The rules of this system are:
\begin{itemize}
    \item \emph{the weakening rule}: from a linear clause we can derive any other linear clause which is semantically implied,
    \item \emph{the resolution rule}: for every two linear clauses $C$ and $D$ and linear form $f$, we can derive $C \vee D$ from $(f = 0) \vee C$ and $(f= 1) \vee D$.
\end{itemize}
They proved exponential lower bounds for the tree-like restriction of this system. These lower bounds were later extended in \cite{DBLP:conf/csr/Gryaznov19,DBLP:journals/cc/PartT21}. For DAG-like proofs, the only known results are due to Khaniki \cite{DBLP:journals/eccc/Khaniki20} who proved almost quadratic lower bounds, and to Lauria \cite{DBLP:journals/ipl/Lauria18a} for a restriction of the system when parities are on a bounded number of variables. Super-polynomial lower bounds for unrestricted DAG-like $\reslin$ are widely open.

\bigskip
\noindent{\bf Parity decision trees and tree-like $\reslin$.} Given an unsatisfiable CNF $F = C_1 \wedge \ldots \wedge C_m$, the \emph{search problem for $F$} is the computational problem of finding a clause $C_i$ falsified by a given assignment to the variables. A tree-like $\reslin$ refutation of $F$ can be viewed as a parity decision tree solving the search problem for an unsatisfiable CNF \cite{DBLP:conf/csr/Gryaznov19}. Recall that the strongest average-case lower bounds for $\ac0[2]$ are in fact for parity decision trees. Thus it seems that parity decision trees are at the frontier of our understanding in these two areas. Therefore a natural approach to make progress towards both general $\reslin$ lower bounds and average-case hardness for $\ac0[2]$ is to consider DAG-like structures more general than decision trees.

\subsection{Our contributions}

Motivated by strengthening tree-like $\reslin$ lower bounds as well as average-case lower bounds for parity decision trees to more general models, we consider a model of read-once branching programs (BPs) with linear queries. The most natural way to interpret the property of being read-once in BPs with linear queries, is to impose that along every path, the queries are linearly independent. We consider two restrictions of this model which we call \wreadonce{} and \sreadonce{}, both of which extend parity decision trees as well as standard read-once branching programs.

% which extends both parity decision trees and standard read-once branching programs. 

%In this paper we focus on the computational aspect of this model and leave the corresponding proof complexity investigation to future work.

For \sreadonce{} BPs, we prove average-case hardness for a new class of psuedo-random functions, and we give an explicit construction of such a function, thus strengthening the result of Cohen and Shinkar \cite{DBLP:conf/innovations/CohenS16} and making progress towards average-case hardness for $\dnf \circ \Mod_2$. Our pseudo-random functions are defined below and might be of independent interest.

\bigskip
\noindent{\bf Directional affine extractors.} The average-case hardness result of Cohen and Shinkar \cite{DBLP:conf/innovations/CohenS16} is for affine extractors. An affine extractor for dimension $d$ and bias $\epsilon$ is a function such that restricted to any affine subspace of dimension at least $d$ it has bias at most $\epsilon$. Explicit constructions for such functions are known (e.g., \cite{MR2306652,DBLP:journals/combinatorica/Yehudayoff11,DBLP:journals/siamcomp/Ben-SassonK12}). For our purposes it is not clear if affine extractors are sufficient. Therefore we consider a more robust concept. We say that a function $f\colon \Q^n \rightarrow \Q^n$ is a \emph{directional affine extractor} for dimension $d$ with bias $\epsilon$, if for every non-zero $a \in \Q^n$, the derivative of $f$ in the direction $a$, $D_a f(x) = f(x+a) + f(x)$, is an affine extractor for dimension $d$ with bias $\epsilon$. We give an explicit construction of a good directional affine extractor for dimension larger than $2n/3$.

For \wreadonce{} BPs we show a correspondence with $\reslin$. More precisely, we show that a \wreadonce{} BP solving the search problem for a CNF $F$, can be converted to a $\reslin$ refutation of $F$. This also justifies considering a $\reslin$ counterpart to regular resolution. Recall that in a regular resolution proof, no variable is resolved more than once along any path. It is well-known that a read-once BP solving the search problem for an unsatisfiable CNF can be converted to a regular resolution refutation of the formula. Our result should be interpreted as an extension of this result to $\reslin$.

%Our contributions can be summarized as follows.

%\begin{itemize}
%\item An intuitive definition of \readonce{} linear branching programs, where queries are linear forms.
%\item Criterion
%\end{itemize}

\subsection{\Readonce{} linear branching programs}

The model of read-once branching programs is a natural and extensively studied model of computation for which strong lower bounds are known~\cite{DBLP:journals/tcs/SavickyZ00,DBLP:conf/icalp/AndreevBCR99}. Here we consider an extension of this model where queries are linear forms. A linear branching program $\calP$ in the variables $x$ is a DAG with the following properties:
\begin{itemize}
  \item it has exactly one source;
  \item it has two sinks labeled with $0$ and $1$ representing the values of the function that $\calP$ computes;
  \item every inner node is labeled by a linear form $q$ over $\bbF_2$ in $x$ which we call \emph{queries};
  \item every inner node with a label $q$ has two outgoing edges labeled with $0$ and $1$ representing the value of $q$.
\end{itemize}

Any assignment to the input variables naturally defines a path in the program. We say that $\calP$ computes a Boolean function $f\colon \Q^n \to \Q$ if for every $x \in \Q^n$, the path in $\calP$ defined by $x$ ends in the sink labeled with $f(x)$.

We now define \readonce{} linear BPs. Given an inner node $v$ of a linear branching program $\calP$, we define $\pre_v$ as the span of all queries that appear on any path from the source of $\calP$ to $v$, excluding the query at $v$. We define $\post_v$ as the span of all queries in the subprogram starting at $v$. 

\begin{definition}[\Wreadonce{} linear branching program]
We say that a linear branching program $\calP$ is \emph{\wreadonce} if for every inner node $v$ of $\calP$ which queries $q$, it holds that $q \not \in \pre_v$.
\end{definition}

We can make this requirement more strict.

\begin{definition}[\Sreadonce{} linear branching program]
We say that a linear branching program $\calP$ is \emph{\sreadonce} if for every inner node $v$ of $\calP$, it holds that $\pre_v \cap \post_v = \{0\}$.
\end{definition}

It follows from both definitions that queries alongside any path in weakly or strongly read-once BP are linearly independent. Furthermore, both of these models generalize standard read-once BPs and parity decision trees.

When the distinction between weakly or strongly read-once is not important, we simply write ``read-once''.

\section{Notation and basic facts}

Each path in a \readonce{} program defines an affine subspace given by the set of solutions of the system corresponding to the queries on the path. Any affine subspace can be represented by a vector space shifted by a vector from the affine space. For our purposes, we need to choose this shift carefully. 

Let $p$ be a path in $\calP$ leading to a node $v$ with queries $q_1, \ldots, q_k$ and answers $a_1, \ldots, a_k$ to these queries which define the affine subspace $S_p = \{x:\bigwedge_{i=1}^k q_i(x) = a_i\}$. Let $V_p$ be the supporting vector space of $S_p$, i.e., $V_p = \{x : \bigwedge_{i = 1}^k q_i(x) = 0\}$. Then clearly $S_p = V_p + b$ for any $b \in S_p$. Choose an arbitrary basis $q'_1, \ldots, q'_t$ for $\post_v$. Since $q'_1, \ldots, q'_t$ are independent of $q_1, \ldots, q_k$, there exists $b$ such that $\bigwedge_{i=1}^k q_i(b) = a_i$ and $\bigwedge_{i = 1}^t q'_i(b) = 0$. Then $S_p = V_p + b$ and for every $q \in \post_v$, we have $q(b) = 0$.

\begin{definition}[Canonical affine subspace] 
Given a path $p$ which ends at a node $v$, we call $S_p$ the \emph{canonical affine subspace} for $p$. Furthermore a \emph{canonical representation} of $S_p$ is any $V_p + b = S_p$ where every $q \in \post_v$ vanishes on $b$.
%Let $p$ be a path in $\calP$ of length $k$. Path $p$ consists of queries (linear forms) $q_1, \ldots, q_k$ and answers $a_1, \ldots, a_k$ to these queries, and defines the affine subspace $\bigwedge_{i=1}^k q_i(x) = a_i$. We will define a canonical representation of this affine subspace. Let $V$ be the solutions of the linear system $\bigwedge_{i=1}^k q_i(x) = 0$. Since $q_1,\ldots,q_k \in \pre_v$, we can choose $b$ such that $\bigwedge_{i=1}^k q_i(b) = a_i$ and $b$ vanishes on $\post_v$. Then the input $x \in \Q^n$ passes through $p$ if and only if $x \in V+b$. In addition, for every linear form $q \in \post_v$, $q(x) = 0$ for $x \in V+b$. We call $V+b$ a \emph{canonical affine subspace} for the path $p$.
\end{definition}

Throughout the paper we drop the word \emph{representation} and simply say $V_p + b$ is the canonical affine subspace of $p$
 to mean that it is a canonical representation of $S_p$.
%Note that a canonical affine subspace may not be unique if the linear system that defines $b$ is underdetermined.

Since we will often use canonical affine subspaces to represent paths in BPs, we adopt the following algebraic notation.
Let us denote the space of all linear forms on $\bbF_2^n$ (the dual space) as $\dual{(\bbF_2^n)}$.
Given a subspace $V$ of $\bbF_2^n$, we define $\orth{V}$ as the space of all linear forms from $\dual{(\bbF_2^n)}$ that vanish on $V$ (this space is sometimes called the annihilator of $V$), i.e.,
\[
  \orth{V} = \{\ell \in \dual{(\bbF_2^n)} : \forall v \in V,\ \ell(v) = 0\}.
\]
Given a path $p$ with queries $q_1, \ldots, q_k$ and its canonical affine subspace $V+b$, the space $\orth{V}$ is the query space of $p$, i.e., $\orth{V} = \Span(q_1, \ldots, q_k)$.

Throughout the paper we adopt the following notation.
\begin{itemize}
  \item Given a vector $c \in \Q^n$ the \emph{support of $c$} is defined as
  \[
    \supp(c) \coloneqq \{i : c_i \neq 0\}.
  \]

  \item Let $\sigma$ be a partial assignment to the variables $x_1, \ldots, x_n$. Then
  \[
    \dom(\sigma) \coloneqq \{i : \sigma(x_i)\ \text{is defined}\}.
  \]

  \item We say that $a \in \Q^n$ is \emph{consistent} with a partial assignment $\sigma$ to $x_1, \ldots, x_n$ if for every $i \in \dom(\sigma)$, it holds that $\sigma(x_i) = a_i$.
  
  \item Let $V$ and $W$ be two subspaces. Then the sum of $V$ and $W$ is the subspace
  \[
    V+W \coloneqq \{v+w : v \in V,\ w \in W\}.
  \]
  Note that $V + W = \Span(V \cup W)$.

  \item We write $a+b$ without specifying the underlying field, if it is clear from the context and often intended to be $\bbF_2$.
\end{itemize}

\subsection{The trace map}

The trace map $\Tr\colon \bbF_{p^n} \to \bbF_p$ is defined as
\[
  \Tr(x) \coloneqq \sum_{i=0}^{n-1} x^{p^i}.
\]
One important property that we need is that $\Tr$ is a linear map. We also use the following fact about the trace.
\begin{proposition}[cf.~\cite{lidl_niederreiter_1996}]\label{prop:trace_linear}
  For every $\bbF_p$-linear map $\pi\colon \bbF_{p^n} \to \bbF_p$ there exists $\mu \in \bbF_{p^n}$ such that for all $x \in \bbF_{p^n}$ we have
  \[
    \pi(x) = \Tr(\mu \cdot x).
  \]
  Furthermore, $\pi$ is trivial if and only if $\mu = 0$.
\end{proposition}

Since, we are interested in Boolean functions, we will only consider the case $p=2$.
Let $\phi\colon \bbF_2^n \to \bbF_{2^n}$ be any $\bbF_2$-linear isomorphism. Then $\Tr(\mu \cdot \phi(x))$ is linear in $x$ and we have the following:
\begin{proposition}\label{prop:trace_linear_boolean}
  The set of all linear Boolean functions coincides with the set of functions $\ell_\mu(x) = \Tr(\mu \cdot \phi(x))$, where $\mu \in \bbF_{2^n}$.
\end{proposition}

In the rest of the paper we fix $\phi$. To make the proofs more readable we use bold font to denote the corresponding elements of $\bbF_{2^n}$, e.g., $\iso{x}$ for $\phi(x)$.

\subsection{Affine extractors and dispersers}
A Boolean function $f\colon \Q^n \to \Q$ is an \emph{affine disperser} for dimension $d$ if $f$ is not constant on any affine subspace of dimension at least $d$.
Let us also recall affine extractors, which are generalizations of affine dispersers.

The \emph{bias} of $f$ is defined as
\[
  \bias(f) \coloneqq \card{\E_{x \in U_n}[(-1)^{f(x)}]},
\]
where $U_n$ is a uniform distribution on $\Q^n$.
Given an affine subspace $f$, \emph{the bias of $f$ restricted to $S \subseteq \Q^n$} is defined as
\[
  \bias(\restr{f}{S}) \coloneqq \card{\E_{x \in U(S)}[(-1)^{f(x)}]},
\]
where $U(S)$ is a uniform distribution on $S$.

A Boolean function $f\colon \Q^n \to \Q$ is an \emph{affine extractor} for dimension $d$ with bias $\epsilon$ if for every affine subspace $S$ of dimension $d$, the bias of $f$ restricted to $S$, $\bias(\restr{f}{S})$, is at most $\epsilon$.

\section{Affine mixedness}

In this section we give a criterion for functions to be worst-case hard for \readonce{} linear BPs. Let us first recall \emph{mixedness} from standard read-once BPs.

\begin{definition}
  A Boolean function $f\colon \Q^n \to \Q$ is \emph{$d$-mixed} if for every $I \subseteq [n]$ of size at most $n-d$\footnote{This definition is commonly given for sets of size $d$ instead of $n - d$. We deviate from this since for our generalization to affine spaces, it corresponds to dimension which is more natural.} and every two distinct partial assignments $\sigma$ and $\tau$ with $\dom(\sigma) = \dom(\tau) = I$, it holds that $\restr{f}{\sigma} \neq \restr{f}{\tau}$.
\end{definition}

\begin{theorem}[Folklore; see~\cite{DBLP:books/daglib/0028687} for a proof]\label{thm:mixedness_1bp}
  Let $f\colon \Q^n \to \Q$ be a $d$-mixed Boolean function. Then any read-once branching program computing $f$ has size at least $2^{n-d}-1$.
\end{theorem}

Explicit constructions of $d$-mixed functions with $d = o(n)$ and thus $2^{n - o(n)}$ size lower bounds for read-once BPs were given in \cite{DBLP:journals/tcs/SavickyZ00,DBLP:conf/icalp/AndreevBCR99}. We generalize this notion for linear branching programs. We need the following equivalent definition of $d$-mixedness.

\begin{lemma}\label{lm:equiv_d-mixed}
  A Boolean function $f$ is $d$-mixed if and only if for every partial assignments $\sigma$ of size at most $n-d$ and every $c \neq 0$ with $\supp(c) \subseteq \dom(\sigma)$, there exists $x$ consistent with $\sigma$ such that $f(x) \neq f(x+c)$.
\end{lemma}
\begin{proof}
  ($\Leftarrow$) Let $\sigma$ and $\tau$ be two distinct partial assignments with domain $I$ of size at most $n - d$. Define $c_i = \tau(x_i) + \sigma(x_i)$ for $i \in I$ and $c_i = 0$ otherwise. By assumption there exists $x$ consistent with $\sigma$ such that $f(x) \ne f(x + c)$. It follows from the definition of $c$ that $x+c$ is consistent with $\tau$. Define $J = [n] \setminus I$ and $z = x_{J} = {(x+c)}_J$. Then $\restr{f}{\sigma}(z) = f(x) \neq f(x+c) = \restr{f}{\tau}(z)$.

  ($\Rightarrow$) Let $\sigma$ be a partial assignment with a domain of size at most $n - d$ and let $c$ be given such that $\supp(c) \subseteq \dom(\sigma)$. Define $\tau(x_i) = \sigma(x_i) + c_i$ for $i \in \dom(\sigma)$. By assumption $\restr{f}{\sigma} \neq \restr{f}{\tau}$, hence there exists $z$ such that $\restr{f}{\sigma}(z) \neq \restr{f}{\tau}(z)$. Define $x$ to take the same value as $\sigma$ on $\dom(\sigma)$ and equal to $z$ otherwise. Then $f(x) = \restr{f}{\sigma}(z) \neq \restr{f}{\tau}(z) = f(x+c)$.
\end{proof}

\begin{definition}
  A Boolean function $f\colon \Q^n \to \Q$ is \emph{$d$-affine mixed} if for every affine subspace $S$ of dimension at least $d$ and every vector $c \not\in V$, where $V$ is the supporting vector space of $S$, there exists $x \in S$ such that $f(x) \neq f(x+c)$.
\end{definition}

It follows from \cref{lm:equiv_d-mixed} that $d$-affine mixedness implies $d$-mixedness since a partial assignment is a special case of an affine subspace.

Now we are ready to prove a generalization of \cref{thm:mixedness_1bp}.

\begin{theorem}
  Let $f\colon \Q^n \to \Q$ be a $d$-affine mixed Boolean function. Then any \sreadonce{} linear branching program computing $f$ has size at least $2^{n-d}-1$.
\end{theorem}
\begin{proof}
  We prove that any such program $\calP$ computing $f$ starts with a complete binary tree of depth $n-d-1$. Assume for the sake of contradiction that there are two paths $p$ and $q$ of length at most $n-d-1$, which meet for the first time at a node $v$. Let $V+a$ and $W+b$ be their corresponding canonical affine subspaces. Both of them have dimension at least $d+1$.

  We start by proving $\orth{V} = \orth{W}$ which implies $V = W$. Suppose that it is not the case. Without loss of generality, there exists $\ell \in \orth{W} \setminus \orth{V}$. By the \readonce{} property $\ell \not\in \post_v$.

  Consider two affine subspaces $V'+a_1$ and $V'+a_2$ obtained by intersecting $V+a$ with $\ell(x) = 0$ and $\ell(x)=1$ such that for every $\ell' \in \post_v$, $\ell'(a_1) = 0$ and $\ell'(a_2) = 0$ (recall that we can choose such $a_1$ and $a_2$ since $\pre_v \cap \post_v = \{0\}$). By construction, they have dimension at least $d$.
  Since $f$ is $d$-affine mixed, there exists $z \in V'$ such that $f(z+a_1) \neq f(z+a_2)$.
  Consider any query $\ell'$ in the subprogram starting at $v$. The fact that $\ell' \in \post_v$ implies $\ell'(a_1) = \ell'(a_2) = 0$. Thus, we have $\ell'(z+a_1) = \ell'(z) = \ell'(z+a_2)$. It implies that in the subprogram starting at $v$ both $z+a_1$ and $z+a_2$ must follow the same path contradicting $f(z+a_1) \neq f(z+a_2)$.

  Now, since $V = W$, $V+b$ is the canonical affine subspace for $q$, and $a \ne b$ since $p$ and $q$ are different paths.
  Again, since $f$ is $d$-affine mixed, there exists $z \in V$ such that $f(z+a) \neq f(z+b)$. Analogously to the previous case, for every $\ell' \in \post_v$ we have $\ell'(a)=\ell'(b)=0$, and thus $\ell'(z+a) = \ell'(z+b)$ contradicting $f(z+a) \neq f(z+b)$.
\end{proof}

\section{Affine dispersers for directional derivatives}

In this section we give an explicit construction of an affine mixed function for linear dimension. In fact, we give an even more powerful construction, which allows us to get an average-case lower bound for \sreadonce{} linear branching programs.

For a Boolean function $f$ its directional derivative with respect to a non-zero vector $a$ is defined as
\[
  D_{a} f(x) \coloneqq f(x+a) + f(x).
\]

\begin{definition}
  A Boolean function $f\colon \Q^n \to \Q$ is a \emph{directional affine extractor} for dimension $d$ with bias $\epsilon$ if for every non-zero $a$, the derivative $D_a f$ is an affine extractor for dimension $d$ with bias $\epsilon$.

  Similarly, $f$ is a \emph{directional affine disperser} for dimension $d$ if for every non-zero $a$, $D_a f$ is an affine disperser for dimension $d$.
\end{definition}

Observe that this notion is stronger than the one defined in the previous section: if $f$ is a directional affine disperser for dimension $d$, then it is $d$-affine mixed. %We also point out that a directional affine extractor is in particular an affine extractor.

%\begin{proposition}
%If $f$ is a directional affine extractor for dimension $d$ and bias $\epsilon$, then it is also an affine extractor for dimension $d+1$ with bias $\epsilon$.
%\end{proposition}

%\begin{proof}
%Let $S$ be an affine subspace of dimension $d+1$. Write $S= S_0 \cup S_1$ where $S_0$ and $S_1$ are two disjoint $d$-dimensional affine subspaces. We write $S_0 = V + a$ and $S_1 = V + b$ for some vector space $V$ and $a \ne b$. We can bound $\bias(\restr{f}{S})$ as follows.
%\begin{eqnarray*}
%\left|\E_{x \in U(S)}[(-1)^{f(x)}]\right| &=& \frac{1}{2}\left|\E_{x \in U(S_0)}[(-1)^{f(x)}] + \E_{x \in U(S_1)}[(-1)^{f(x)}]\right| \\
%&=& \frac{1}{2}\card{\E_{x \in U(V+a)}[(-1)^{f(x)}] + \E_{x \in U(V+b)}[(-1)^{f(x)}]}\\
%&=& \frac{1}{2}\card{\E_{x \in U(V+a)}[(-1)^{f(x)} + (-1)^{f(x+a+b)}]}\\
%&=& \card{\E_{x \in U(V + a)}[(-1)^{f(x)+f(x + a + b)}]}\\
%&\le& \epsilon,
%\end{eqnarray*}
%where the inequality follows by assumption that $D_{a+b}f$ is an affine extractor for dimension $d$ with bias $\epsilon$.
%\end{proof}

In what follows we construct a Boolean function $f$ in $n$ variables that is a good directional affine extractor for dimensions bigger than $\frac{2}{3}n$.

%The starting point is the inner product function. Let $n \ge 2$ be an even integer. The inner product function $\IP$ is defined as follows:
%\[
%  \IP(x_1, \ldots, x_n) \coloneqq \sum_{i=1}^{n/2} x_i x_{n/2+i} \pmod 2.
%\]

It is a well-known fact that the inner product function is an affine extractor. IP is a member of the class of bent functions, which are all affine extractors. A Boolean function $f\colon \Q^n \to \Q$ is called a \emph{bent function} if all Fourier coefficients of its $\pm 1$ representation $f_\pm(x) \coloneqq {(-1)}^{f(x)}$ have the same absolute value.
%\todo{Briefly state some applications here:~\cite{DBLP:conf/innovations/CohenS16,DBLP:journals/jcss/CheraghchiGJWX18}, Viola?}
%Cohen and Shinkar~\cite{DBLP:conf/innovations/CohenS16} used this result for the inner product function $\IP$ to obtain average-case lower bounds for parity decision trees.

\begin{lemma}[Folklore; for a proof see, e.g.,~\cite{DBLP:conf/innovations/CohenS16,DBLP:journals/jcss/CheraghchiGJWX18}]
\label{lm:bent}
  Let $f$ be a bent function on $n$ variables and $c \ge 1$ be an integer. Then, $f$ is an affine extractor for dimension $k = n/2 + c$ with bias at most $2^{-c}$. In particular, $f$ is an affine disperser for dimension $n/2 + 1$.
\end{lemma}

We apply this result to prove that the following function is an affine extractor.
\begin{lemma}\label{lem:trace_affine_extractor}
  Let $a_0, a_1, a_2, a_3 \in \bbF_{2^k}$ with $a_0 \ne 0$. Let $g\colon \Q^k \times \Q^k \to \Q$ be the function defined as
  \[
    g(x, y) = \Tr(a_0 \cdot \phi(x) \cdot \phi(y) + a_1 \cdot \phi(x) + a_2 \cdot \phi(y) + a_3).
  \]
  Then $g$ is an affine extractor for dimension $k+c$ with bias at most $2^{-c}$.
  In particular, $g$ is an affine disperser for dimension $k+1$.
\end{lemma}
\begin{proof}
  Let $g_\pm$ be the $\pm 1$ representation of $g$. By Lemma \ref{lm:bent}, it is enough to prove that all Fourier coefficients of $g_\pm$ have the same absolute value. Recall that given $\alpha \in \Q^{2k}$ the $\alpha$-character $\chi_\alpha$ is defined as $\chi_\alpha(x, y) = {(-1)}^{\alpha \cdot (x, y)}$, where $\alpha \cdot (x, y)$ is the inner product. Fourier coefficient $\widehat{g_\pm}(\alpha)$ can be computed as follows.
  \[
    \widehat{g_\pm}(\alpha) = \sum_{x,y \in \Q^k} g_\pm(x, y) \cdot \chi_\alpha(x, y) = \sum_{x,y \in \Q^k} {(-1)}^{\Tr(a_0 \cdot \iso{x} \cdot \iso{y} + a_1 \cdot \iso{x} + a_2 \cdot \iso{y} + a_3)} \cdot \chi_\alpha(x, y).
  \]

    Split $\alpha$ into two equal parts: $\alpha = (\alpha_1, \alpha_2)$. Then $\alpha \cdot (x, y) = \alpha_1 \cdot x + \alpha_2 \cdot y$.
  By \cref{prop:trace_linear_boolean}, there exist $\mu_1, \mu_2 \in \bbF_{2^k}$ such that $\alpha_1 \cdot x = \Tr(\mu_1 \cdot \iso{x})$ and $\alpha_2 \cdot y = \Tr(\mu_2 \cdot \iso{y})$. Also define
  \begin{align*}
    b_1 &\coloneqq a_0^{-1} \cdot (a_1 + \mu_1), \\
    b_2 &\coloneqq a_0^{-1} \cdot (a_2 + \mu_2), \\
    b_3 &\coloneqq a_3 + a_0^{-1} \cdot (a_1 + \mu_1) \cdot (a_2 + \mu_2).
  \end{align*}

  Then we can express $\widehat{g_\pm}(\alpha)$ in terms of $b_i$:
  \begin{align*}
    \widehat{g_\pm}(\alpha) &=
    \sum\limits_{x,y \in \Q^k} {(-1)}^{\Tr(a_0 \cdot \iso{x} \cdot \iso{y} + a_1 \cdot \iso{x} + a_2 \cdot \iso{y} + a_3)} \cdot {(-1)}^{\Tr(\mu_1 \cdot \iso{x}) + \Tr(\mu_2 \cdot \iso{y})} \\ &=
    \sum\limits_{x,y \in \Q^k} {(-1)}^{\Tr(a_0 \cdot \iso{x} \cdot \iso{y} + a_1 \cdot \iso{x} + a_2 \cdot \iso{y} + a_3 + \mu_1 \cdot \iso{x} + \mu_2 \cdot \iso{y})} \\ &=
    \sum\limits_{x,y \in \Q^k} {(-1)}^{\Tr(a_0 \cdot (\iso{x} + b_2) \cdot (\iso{y} + b_1) + b_3)} \\ &=
    {(-1)}^{\Tr(b_3)} \cdot \sum\limits_{x,y \in \Q^k} {(-1)}^{\Tr(a_0 \cdot (\iso{x} + b_2) \cdot (\iso{y} + b_1))}.
  \end{align*}

  Since $x$ and $y$ iterate through all vectors from $\Q^k$, $a_0 \cdot (\iso{x} + b_2)$ and $\iso{y} + b_1$ take all possible values from $\bbF_{2^k}$. It follows that
  \[
    \widehat{g_\pm}(\alpha) = {(-1)}^{\Tr(b_3)} \widehat{g_\pm}(0).
  \]
\end{proof}

We are now ready to present our directional affine extractor.

\begin{theorem}
\label{thm:construction}
  Let $f\colon \Q^k \times \Q^k \times \Q^k \to \Q$ be the function defined by
  \[
    f(x, y, z) = \Tr(\phi(x) \cdot \phi(y) \cdot \phi(z)).
  \]
  Then $f$ is a directional affine extractor for dimension $2k+c$ with bias $\epsilon \le 2^{-c}$. In particular, $f$ is a directional affine disperser for dimension $2k+1$.
\end{theorem}
\begin{proof}
  Consider the directional derivative of $f$ in the non-zero direction $a = (a_1,a_2,a_3)$:
  \begin{align*}
    D_{a} f(x, y, z) &=
    f(x+a_1, y+a_2, z+a_3) + f(x, y, z) \\ &=
    \Tr(\phi(x+a_1) \cdot \phi(y+a_2) \cdot \phi(z+a_3)) + \Tr(\iso{x} \cdot \iso{y} \cdot \iso{z}).
  \end{align*}

  By linearity of $\Tr$ and $\phi$ we have
  \begin{align*}
    D_{a} f(x, y, z) &=
    \Tr((\iso{x}+\iso{a_1}) \cdot (\iso{y} + \iso{a_2}) \cdot (\iso{z} + \iso{a_3}) + \iso{x} \cdot \iso{y} \cdot \iso{z}) \\ &=
    \Tr(\iso{a_1} \cdot \iso{y} \cdot \iso{z} + \iso{a_2} \cdot \iso{x} \cdot \iso{z} + \iso{a_3} \cdot \iso{x} \cdot \iso{y} + \ell(\iso{x},\iso{y},\iso{z})),
  \end{align*}
  where $\ell$ is an affine function.

  Without loss of generality we may assume that $a_3 \neq 0$.
  Let $S$ be an affine subspace with dimension at least $2k+c$. We need show that the bias of $f$ restricted to $S$ is at most $\epsilon$. Given $z_0 \in \Q^k$ define $S_{z_0} \coloneqq \{(x, y) : (x, y, z_0) \in S\}$. For every $z_0$ the affine subspace $S_{z_0}$ is either empty or has dimension at least $k+c$.
  Consider the restriction of $D_a f$ to $z = z_0$.
  \[
    h_{z_0}(x, y) \coloneqq D_a f(x, y, z_0) = \Tr(\iso{a_3} \cdot \iso{x} \cdot \iso{y} + \ell'_{z_0}(\iso{x}, \iso{y})),
  \]
  where $\ell'_{z_0}$ is an affine function. By \cref{lem:trace_affine_extractor}, $h_{z_0}$ is an affine extractor for dimension $k+c$ with bias $\epsilon \le 2^{-c}$. In particular, if $S_{z_0}$ is non-empty, then $\bias(\restr{h_{z_0}}{S_{z_0}}) \le \epsilon$.

  Thus, the bias of $D_a f$ restricted to $S$ can easily be bounded as follows:
  \begin{align*}
    \bias(\restr{D_a f}{S}) &=
    \left|\frac{1}{\card{S}} \sum_{(x, y, z) \in S} {(-1)}^{D_a f(x, y, z)}\right| \\ &=
    \left|\frac{1}{\card{S}} \sum_{z_0 \in \Q^n} \sum_{(x, y, z_0) \in S} {(-1)}^{D_a f(x, y, z_0)}\right| \\ &=
    \left|\frac{1}{\card{S}} \sum_{z_0 \in \Q^n} \sum_{(x, y) \in S_{z_0}} {(-1)}^{h_{z_0}(x, y)}\right| \\ &\le
    \frac{1}{\card{S}} \sum_{z_0 \in \Q^n} \left|\sum_{(x, y) \in S_{z_0}} {(-1)}^{h_{z_0}(x, y)}\right| \\ &\le
    \frac{1}{\card{S}} \sum_{z_0 \in \Q^n} \epsilon \cdot \card{S_{z_0}} = \epsilon.
  \end{align*}
\end{proof}

\section{Average-case lower bound}

We consider a canonical form of \sreadonce{} linear branching programs. We adopt the terminology of~\cite{DBLP:journals/cc/ChenKKSZ15} and say that a \readonce{} linear branching program is \emph{full} if for every inner node $v$ of the program, all the paths leading to $v$ have the same query space.

A \emph{multipath} $(w_1, \ldots, w_m, v)$ is a linear branching program of the form
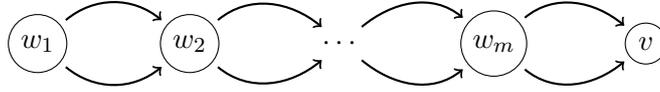
\begin{center}
  \begin{tikzpicture}
    \tikzset{pnode/.style={draw,circle,outer sep=1mm,inner sep=1mm}}
    \tikzset{pedge/.style={thick,->}}
    \node[pnode] (u1) at (0, 0) {$w_1$};
    \node[pnode] (u2) at (2, 0) {$w_2$};
    \node (ellipsis) at (4, 0) {$\cdots$};
    \node[pnode] (um) at (6, 0) {$w_m$};
    \node[pnode] (v) at (8, 0) {$v$};
    \draw[pedge] (u1) to [out=40,in=140] (u2);
    \draw[pedge] (u1) to [out=-40,in=-140] (u2);
    \draw[pedge] (u2) to [out=40,in=140] (ellipsis);
    \draw[pedge] (u2) to [out=-40,in=-140] (ellipsis);
    \draw[pedge] (ellipsis) to [out=40,in=140] (um);
    \draw[pedge] (ellipsis) to [out=-40,in=-140] (um);
    \draw[pedge] (um) to [out=40,in=140] (v);
    \draw[pedge] (um) to [out=-40,in=-140] (v);
  \end{tikzpicture}
\end{center}

That is, the program ignores the answers to the queries at $w_i$ for every $i$. Given a program $\calP$, we say that a subset of nodes is an \emph{antichain} if none of its nodes is a descendant of another. For example, the set of nodes at a fixed depth and the set of leaves form an antichain. The following lemma and its proof are easy extensions of Lemma 3.7 in \cite{DBLP:journals/cc/ChenKKSZ15}.

\begin{lemma}\label{lem:full_bp}
  Every \wreadonce{} or \sreadonce{} linear branching program $\calP$ of size $s$ in $n$ variables has an equivalent full \wreadonce{} or \sreadonce{} linear branching program $\calP'$, respectively, of size at most $3n \cdot s$. Furthermore, the size of every antichain in $\calP'$ is at most $2s$.
\end{lemma}
\begin{proof}
  We construct $\calP'$ inductively. Consider the nodes of $\calP$ in topological order. It is clear that the start node satisfies the fullness property. Let $v$ be a node of $\calP$ and $p_1, \ldots, p_k$ the paths that meet at $v$, and $V_1 + a_1, \ldots, V_k + a_k$ their canonical affine subspaces. %Let $V$ be a vector space such that $\orth{V} = \orth{V_1} + \cdots + \orth{V_k}$. 
  For every $i \in [k]$ choose a set of linearly independent queries $Q_i$ such that $\orth{V_i} + \Span(Q_i) = \pre(v)$.

  For every $i \in [k]$ do the following. Let $Q_i = \{q_1, \ldots, q_m\}$. Replace the edge $u_i \to v$ with a multipath $(w_1, \ldots, w_m, v)$ and an edge $u_i \to w_1$, where $w_i$ are labeled with $q_i$. After this transformation, every path to $v$ will have the query space $\pre(v)$.

  Since a branching program of size $s$ has at most $2s$ edges and we replaced every edge with a multipath of length at most $n$, the size of the constructed full \readonce{} linear branching program $\calP'$ is at most $s+2s \cdot n \le 3n \cdot s$.

  Consider an antichain $A$ in $\calP'$. We map every node in $A$ to nodes in $\calP$. Each node in $A$ is either originally in $\calP$ or it was created by a multipath. In the former case we map it to itself, and in the latter case we map it to the parent node from which it was created. Since the out-degree in $\calP$ is 2 and $A$ is an antichain, at most 2 nodes are mapped to the same node. This proves the result.
\end{proof}

Denote by $\dist(f,g)$ the relative distance between Boolean function $f$ and $g$.

\begin{theorem}\label{thm:avg_case_full}
  Let $f\colon \Q^n \to \Q$ be a directional affine extractor for dimension $d$ with bias $\epsilon < \frac{1}{2}$. Then for every $g\colon \Q^n \to \Q$ computed by a \sreadonce{} linear branching program $\calP$ of size at most $\epsilon \cdot 2^{n-d-1}$, it holds that $\dist(f, g) \ge \frac{1-\sqrt{2\epsilon}}{2}$.
\end{theorem}
\begin{proof}
  Let $s$ denote the size of $\calP$. We first convert $\calP$ into a full program. By \Cref{lem:full_bp}, the size of every antichain is at most $2s$. We then construct an equivalent program $\calP'$ in which every path has length at least $n-d$. We can achieve this by extending every leaf of low depth by a multipath of an appropriate length. 

  Consider the set $A$ of nodes in $\calP'$ at depth exactly $n-d$. Note that every $v \in A$ is either a node at depth $n - d$ in $\calP$, or it is uniquely defined by a leaf of $\calP$ by a multipath. Thus $A$ is identified by an antichain in $\calP$ and thus $\card{A} \le 2s$.

  We call an input $x$ \emph{wrong} if $f(x) \neq g(x)$. The distance $\dist(f, g)$ between $f$ and $g$ is the fraction of wrong inputs.

  \begin{claim}\label{clm:avg_case_1}
    Let $v \in A$ and $k$ the numbers of paths that meet at $v$. Then the number of wrong inputs that pass through $v$ is at least
    \[
      \frac{k\cdot 2^d}{2} \left(1-\sqrt{\epsilon+\frac{1}{k}}\right).
    \]
  \end{claim}
  \begin{claimproof}
    Since the program is full, the corresponding canonical affine subspaces for the paths that meet at $v$ are $V+a_1, \ldots, V+a_k$, for some $d$-dimensional vector space $V$, and distinct $a_1, \ldots, a_k \in \Q^n$.
    Recall that $f$ is a directional affine extractor with bias $\epsilon$. Then for every $i \neq j$, it holds that $D_{a_i + a_j} f = f(x+(a_i + a_j)) + f(x)$ is an affine extractor with bias $\epsilon$, thus
    \begin{equation}\label{eq:extractor_assumption}
      \begin{aligned}
        \left|\sum_{x \in V} {(-1)}^{f(x+a_i)} \cdot {(-1)}^{f(x+a_j)}\right| &=
        \left|\sum_{x \in V + a_j} {(-1)}^{f(x+a_i+a_j) + f(x)}\right| \\ &=
        \bias\mathopen{}\mathclose\bgroup\left(\restr{D_{a_i + a_j} f}{V+a_j}\aftergroup\egroup\right) \cdot \card{V} \le \epsilon \card{V}.
      \end{aligned}
    \end{equation}

    Every $x \in V$ produces a partition of $[k]$ into two parts $(J, [k] \setminus J)$ such that $f(x+a_i) = 0$ for $i \in J$ and $f(x+a_i) = 1$ for $i \not\in J$. Let $m_x$ be the size of the \emph{smallest} part.
    By definition of canonical affine subspace and the choice of $a_i$, for any linear query $q \in \post_v$ we have $q(a_i) = 0$ for all $i \in [k]$. Then $x+a_1, \ldots, x+a_k$ will follow the same path in the subprogram starting at $v$. Hence, for every $x \in V$ it holds that $f(x+a_1) = \cdots = f(x+a_k)$. It implies that at least $m_x$ inputs of the form $x+a_i$ are wrong and the total number of wrong inputs passing through $v$ is at least
    \[
      m \coloneqq \sum_{x \in V} m_x.
    \]
  
    Now consider the following sum
    \[
      % \sum_{1\le i<j \le k} \sum_{x \in V + a_j} [f(x+(a_i + a_j)) \neq f(x)] =
      E \coloneqq
      \sum_{\substack{x \in V \\ 1\le i<j \le k}} |f(x+a_i) - f(x+a_j)|.
    \]

    We apply double counting to this quantity to obtain the result.
    On the one hand, by definitions of $m_x$ and $m$, we have
    \begin{equation*}
      E = \sum_{x \in V} m_x \cdot (k - m_x) = k m - \sum_{x \in V} m_x^2.
    \end{equation*}
  
    By the Cauchy--Schwarz inequality, $\sum_{x \in V} m_x^2 \ge {\left(\sum_{x \in V} m_x\right)}^2 / \card{V} = m^2/\card{V}$. Thus,
    \begin{equation}\label{eq:bad_inputs_upper_bound}
      E \le km-m^2/\card{V}.
    \end{equation}
  
    On the other hand, $E$ can be rewritten as follows.
    \begin{equation*}
      \begin{aligned}
        % \eqref{eq:bad_inputs_qudratic} =
        E &=
        \sum_{\substack{x \in V \\ 1\le i<j \le k}} \frac{1}{4} {\left({(-1)}^{f(x+a_i)} - {(-1)}^{f(x+a_j)}\right)}^2 \\ &=
        \frac{1}{4} \sum_{1 \le i < j \le k} \left( 2\card{V} - 2\sum_{x\in V} {(-1)}^{f(x+a_i)} \cdot {(-1)}^{f(x+a_j)}\right).
      \end{aligned}
    \end{equation*}

    Applying~\eqref{eq:extractor_assumption}, we obtain the following lower bound on $E$.
    \begin{equation}\label{eq:bad_inputs_lower_bound}
      E \ge \frac{1}{2} \binom{k}{2} \card{V} (1-\epsilon).
    \end{equation}
  
    Combining~\eqref{eq:bad_inputs_upper_bound} and~\eqref{eq:bad_inputs_lower_bound}, we get
    \[
      km-m^2/\card{V} \ge \frac{1}{2} \binom{k}{2} \card{V} (1-\epsilon).
    \]
  
    This can be written as
    \begin{align*}
      {\left(m-\frac{k\card{V}}{2}\right)}^2 &\le
      \frac{1}{4} k^2 \card{V}^2 - \frac{1-\epsilon}{2}\binom{k}{2}\card{V}^2 \\ &=
      \frac{k^2 \card{V}^2}{4} \left(1 - (1-\epsilon)\left(1-\frac{1}{k}\right)\right) \\ &\le
      \frac{k^2 \card{V}^2}{4} \left(\epsilon + \frac{1}{k}\right).
    \end{align*}
  
    Thus,
    \[
      m \ge \frac{k\card{V}}{2}\left(1 - \sqrt{\epsilon + \frac{1}{k}}\right) = \frac{k\cdot 2^d}{2}\left(1 - \sqrt{\epsilon + \frac{1}{k}}\right).
    \]
  \end{claimproof}

  Let $k(v)$ denote the number of paths that meet at $v$ and define $w(k)$ as
  \[
    w(k) \coloneqq \frac{k 2^d}{2}\left(1-\sqrt{\epsilon+\frac{1}{k}}\right).
  \]
  Then by \cref{clm:avg_case_1} the total number of bad inputs that pass through $A$ is at least $\sum_{v\in A} w(k(v)) = \sum_{v\in A} \frac{k(v) 2^d}{2}\left(1-\sqrt{\epsilon+\frac{1}{k(v)}}\right)$.
  Since all paths in $\calP'$ has length at least $n-d$, $\sum_{v \in A} k(v) = 2^{n-d}$.

  The function $w$ is convex, hence by Jensen's inequality, the total number of bad inputs passing through $A$ is at least
  \begin{align*}
    \sum_{v \in A} w(k(v)) \ge
    \card{A} \cdot w\left(\frac{\sum_{v \in A} k(v)}{\card{A}}\right) = \frac{1}{2} 2^n \left(1-\sqrt{\epsilon + \frac{\card{A}}{2^{n-d}}}\right).
  \end{align*}
  
  Since $\card{A} \le 2s \le \epsilon 2^{n-d}$, this expression is at least
  \[
    \frac{1-\sqrt{2\epsilon}}{2} 2^n.
  \]
\end{proof}

%\Cref{lem:full_bp} immediately implies the following result for general \readonce{} linear branching programs.

%\begin{corollary}\label{cor:avg_case_general}
%  Let $f\colon \Q^n \to \Q$ be a directional affine extractor for dimension $d$ with bias $\epsilon < \frac{1}{2}$. Then for every $g\colon \Q^n \to \Q$ computed by a \readonce{} linear branching program $\calP$ of size at most $\epsilon \cdot 2^{n-d} / 3n$, it holds that $\dist(f, g) \ge \frac{1-\sqrt{2\epsilon}}{2}$.
%\end{corollary}
%{}
%\todo{State that (i) this bound is a bit worse than the worst case; but (ii) the worst case bound requires fewer assumptions about the function.}

Plugging in the function of \Cref{thm:construction} we get the following corollary.

\begin{corollary}
Let $f:\Q^{\frac{n}{3}}\times \Q^{\frac{n}{3}} \times \Q^{\frac{n}{3}} \rightarrow \Q$ be defined by $f(x,y,z) = \Tr(\phi(x) \cdot \phi(y) \cdot \phi(z))$. Then for every $g:\Q^n \rightarrow \Q$ computed by a \sreadonce{} linear BP of size at most $2^{\frac{n}{3} - o(n)}$, $\dist(f, g) \ge \frac{1}{2} - 2^{-o(n)}$.
\end{corollary}

\section{\texorpdfstring{\Wreadonce{} BPs and $\reslin$}{\Wreadonce{} BPs and Res[+]}}

In this section we prove an analogue of the correspondence between read-once BPs and regular resolution for $\reslin$ and \wreadonce{} BPs. The proof is a simple extension of standard arguments.

\begin{theorem}\label{thm:res-bp}\leavevmode
  \begin{enumerate}
      \item Every $\reslin$ refutation of an unsatisfiable CNF $F$ can be translated into a linear BP solving the corresponding search problem without increasing its size.

      \item Every \wreadonce{} BP of size $s$ solving the search problem for CNF $F = C_1 \wedge \ldots \wedge C_m$ in $n$ variables can be translated into a regular $\reslin$ refutation of $F$ of size $O(ns)$.
    %\todo{This statement depends on the definition of regular $\reslin$, since otherwise we won't get a \wreadonce{} program.}
  \end{enumerate}
  % Let $P$ be a \wreadonce{} BP of size $s$ solving the search problem for CNF $F = C_1 \wedge \ldots \wedge C_m$ in $n$ variables. Then there exists a regular $\reslin$ refutation of $F$ of size $O(ns)$.
\end{theorem}

\begin{proof}
   \noindent1.
  Consider an application of the resolution rule in the proof DAG $G$. Suppose that it is applied to clauses $C_0 \lor (f = 0)$ and $C_1 \lor (f = 1)$. Then we label the outgoing edges with $f=1$ and $f=0$ respectively.
  We leave the edges corresponding to the weakening rule unlabeled.

  Let $u$ be a vertex in $G$ and $C_u$ the clause it is labeled with. It can be shown by induction on the depth of $u$ that for every path to $u$, the linear system obtained from the equations written on the edges on this path implies $\neg C_u$. The source contains the empty clause, hence the base case holds. For the inductive step, consider any path leading to $u$ and let $v$ be the parent of $u$ on this path. Consider the case when $v$ corresponds to an application of the resolution rule and $w$ be its other child. Let $C_0 \lor (f=b)$, $C_1 \lor (f=b+1)$, and $C_1 \lor C_2$ be the labels of $u$, $w$, and $v$ respectively, where $b \in \{0,1\}$. By the induction hypothesis, every path to $v$ implies $\neg (C_1 \lor C_2)$. In particular, it implies $\neg C_1$. By construction, the edge $(v, u)$ is labeled with $f=b+1$. Then every path to $u$ going through $v$ implies $\neg C_1 \land (f=b+1) = \neg (C_1 \land (f = b))$. Now consider the case when $u$ corresponds to an application of the weakening rule and let $v$ be it parent on this path. Let $C$ and $D$ be the labels of $u$ and $v$. Every path to $v$ implies $\neg D$ by the induction hypothesis and $\neg D \vDash \neg C$. Thus, every path to $u$ through $v$ implies $\neg C$.

  In particular, every path to the sinks of $G$ falsifies some clause of $F$. To obtain the \wreadonce linear BP, we remove labels at the inner nodes and contract all unlabeled edges.

  \medskip
  \noindent2.
  %It is convenient to think about linear clauses as negations of linear systems.
  A linear clause $C = \bigvee_{i=1}^k (f_i = a_i)$ can be viewed as a negation of a linear system $\neg C = \bigwedge_{i=1}^k (f_i = a_i + 1)$.
  We first convert $P$ into a full BP of size $O(ns)$ using \cref{lem:full_bp}. Inductively, to every node $v$ we associate a linear clause $C_v$ such that:
  \begin{enumerate}
    \item Every assignment reaching $v$ falsifies $C_v$.
    \item If $\neg C_v$ represents a linear system $Bx = b$, then the row space of $B$ is $\pre(v)$.
  \end{enumerate}
  For the base case, with each leaf $v$ we associate the clause $C_v$ it is labeled with. The first condition holds since $P$ solves the search problem. To see the second property, note that any path reaching $v$ can be expressed as a linear system on a basis for $\pre(v)$ which forces every literal in $C_v$. This implies that single variables in $C_v$ are in $\pre(v)$.

  For the inductive step, consider a node $v$, which queries $q$ with outgoing neighbors $u$ and $w$, in the directions $q=0$ and $q=1$ respectively. Observe that $\neg C_u \not \models q(x) = 1$ and $\neg C_w \not \models q(x) = 0$. Thus, there are only two cases to consider:
  \begin{enumerate}
    \item $\neg C_u \not \models q(x)=0$ or $\neg C_w \not \models q(x)=1$,
    \item $\neg C_u \models q(x)=0$ and $\neg C_w \models q(x)=1$.
  \end{enumerate}
  In the first case, we simply let $C_v$ be $C_u$ or $C_w$, depending on which condition holds. For the second case, let $B = \{\beta_1, \ldots, \beta_t\}$ be a basis of $\pre(v)$. Fullness implies $\pre(u) = \pre(w) = \pre(v) + \Span(q)$. Applying the inductive hypothesis, we can write $\neg C_u = (q(x)=0) \wedge (B_u x = b_u)$ and $\neg C_w = (q(x) = 1) \wedge (B_w x = b_w)$, where $B_u$ and $B_w$ are matrices with rows in $\beta_1, \ldots, \beta_t$ and $b_u$ and $b_w$ are some vectors. To write $C_u$ and $C_w$ in these forms, we might need to change the basis, which we can do by applying the weakening rule. We claim that setting $C_v$ so that $\neg C_v$ can be written as $B_u x= b_u \wedge B_w x = b_w$ satisfies the requirements.

  Consider any path to $v$. Such a path can be described by a system $Rx = b$ where rows in $R$ are from $B$. Since every such path can be extended to both $u$ and $w$, it follows that $B_u x = b_u \vDash Rx=b$ and $B_w x = b_w \vDash Rx = b$. This means that $B_u x = b_u \wedge B_w x = b_w$ is consistent and thus the derivation of $C_v$ from $C_u$ and $C_w$ (possibly changing the basis) is a valid $\reslin$ step. It is easy to see that conditions 1 and 2 hold for $C_v$.

  Since for every $v$ we create at most $2$ extra clauses, the total size of the proof is at most $O(ns)$.

\end{proof} 

\section{Conclusion}

Several problems are immediately suggested by our work:

\begin{itemize}
    \item \emph{Explicit constructions.} Give an explicit construction of directional affine extractors (or dispersers) for smaller dimension $d$, ideally $d = o(n)$.

    \item \emph{BP lower bounds.} Prove worst-case and average-case hardness results for the \wreadonce{} BPs.

    \item \emph{Proof complexity.} Prove a \readonce{} linear BP lower bound for a search problem, that is for some unsatisfiable CNF $F = C_1 \wedge \ldots \wedge C_m$, show that a \readonce{} linear BP with leaves labeled by $C_i$s solving the corresponding search problem has to be large.
\end{itemize}

\bibliographystyle{plainurl}
\bibliography{refs}

\begin{thebibliography}{10}

\bibitem{DBLP:conf/innovations/AkaviaBGKR14}
Adi Akavia, Andrej Bogdanov, Siyao Guo, Akshay Kamath, and Alon Rosen.
\newblock Candidate weak pseudorandom functions in {AC0-MOD2}.
\newblock In {\em Innovations in Theoretical Computer Science, ITCS'14,
  Princeton, NJ, USA, January 12-14, 2014}, pages 251--260. {ACM}, 2014.
\newblock \href {https://doi.org/10.1145/2554797.2554821}
  {\path{doi:10.1145/2554797.2554821}}.

\bibitem{DBLP:conf/icalp/AndreevBCR99}
Alexander~E. Andreev, Juri~L. Baskakov, Andrea E.~F. Clementi, and Jos{\'{e}}
  D.~P. Rolim.
\newblock Small pseudo-random sets yield hard functions: New tight explicit
  lower bounds for branching programs.
\newblock In {\em Automata, Languages and Programming, 26th International
  Colloquium, ICALP'99, Prague, Czech Republic, July 11-15, 1999, Proceedings},
  volume 1644 of {\em Lecture Notes in Computer Science}, pages 179--189.
  Springer, 1999.
\newblock \href {https://doi.org/10.1007/3-540-48523-6\_15}
  {\path{doi:10.1007/3-540-48523-6\_15}}.

\bibitem{DBLP:journals/siamcomp/Ben-SassonK12}
Eli Ben{-}Sasson and Swastik Kopparty.
\newblock Affine dispersers from subspace polynomials.
\newblock {\em {SIAM} J. Comput.}, 41(4):880--914, 2012.
\newblock \href {https://doi.org/10.1137/110826254}
  {\path{doi:10.1137/110826254}}.

\bibitem{MR2306652}
Jean Bourgain.
\newblock On the construction of affine extractors.
\newblock {\em Geom. Funct. Anal.}, 17(1):33--57, 2007.
\newblock \href {https://doi.org/10.1007/s00039-007-0593-z}
  {\path{doi:10.1007/s00039-007-0593-z}}.

\bibitem{DBLP:journals/cc/ChenKKSZ15}
Ruiwen Chen, Valentine Kabanets, Antonina Kolokolova, Ronen Shaltiel, and David
  Zuckerman.
\newblock Mining circuit lower bound proofs for meta-algorithms.
\newblock {\em Comput. Complex.}, 24(2):333--392, 2015.
\newblock \href {https://doi.org/10.1007/s00037-015-0100-0}
  {\path{doi:10.1007/s00037-015-0100-0}}.

\bibitem{DBLP:journals/jcss/CheraghchiGJWX18}
Mahdi Cheraghchi, Elena Grigorescu, Brendan Juba, Karl Wimmer, and Ning Xie.
\newblock {\(\mathrm{AC}^0 \circ \mathrm{MOD}_2\)} lower bounds for the boolean
  inner product.
\newblock {\em J. Comput. Syst. Sci.}, 97:45--59, 2018.
\newblock \href {https://doi.org/10.1016/j.jcss.2018.04.006}
  {\path{doi:10.1016/j.jcss.2018.04.006}}.

\bibitem{DBLP:conf/innovations/CohenS16}
Gil Cohen and Igor Shinkar.
\newblock The complexity of {DNF} of parities.
\newblock In {\em Proceedings of the 2016 {ACM} Conference on Innovations in
  Theoretical Computer Science, Cambridge, MA, USA, January 14-16, 2016}, pages
  47--58. {ACM}, 2016.
\newblock \href {https://doi.org/10.1145/2840728.2840734}
  {\path{doi:10.1145/2840728.2840734}}.

\bibitem{DBLP:conf/csr/Gryaznov19}
Svyatoslav Gryaznov.
\newblock Notes on resolution over linear equations.
\newblock In {\em Computer Science - Theory and Applications - 14th
  International Computer Science Symposium in Russia, {CSR} 2019, Novosibirsk,
  Russia, July 1-5, 2019, Proceedings}, volume 11532 of {\em Lecture Notes in
  Computer Science}, pages 168--179. Springer, 2019.
\newblock \href {https://doi.org/10.1007/978-3-030-19955-5\_15}
  {\path{doi:10.1007/978-3-030-19955-5\_15}}.

\bibitem{DBLP:journals/jacm/Hastad21}
Johan H{\aa}stad.
\newblock On small-depth {F}rege proofs for {T}seitin for grids.
\newblock {\em J. {ACM}}, 68(1):1:1--1:31, 2021.
\newblock \href {https://doi.org/10.1145/3425606} {\path{doi:10.1145/3425606}}.

\bibitem{DBLP:conf/focs/HrubesP17}
Pavel Hrube{\v s} and Pavel Pudl{\'{a}}k.
\newblock Random formulas, monotone circuits, and interpolation.
\newblock In {\em 58th {IEEE} Annual Symposium on Foundations of Computer
  Science, {FOCS} 2017, Berkeley, CA, USA, October 15-17, 2017}, pages
  121--131. {IEEE} Computer Society, 2017.
\newblock \href {https://doi.org/10.1109/FOCS.2017.20}
  {\path{doi:10.1109/FOCS.2017.20}}.

\bibitem{DBLP:journals/apal/ItsyksonS20}
Dmitry Itsykson and Dmitry Sokolov.
\newblock Resolution over linear equations modulo two.
\newblock {\em Ann. Pure Appl. Log.}, 171(1), 2020.
\newblock \href {https://doi.org/10.1016/j.apal.2019.102722}
  {\path{doi:10.1016/j.apal.2019.102722}}.

\bibitem{DBLP:books/daglib/0028687}
Stasys Jukna.
\newblock {\em Boolean Function Complexity: Advances and Frontiers}, volume~27
  of {\em Algorithms and combinatorics}.
\newblock Springer, 2012.
\newblock \href {https://doi.org/10.1007/978-3-642-24508-4}
  {\path{doi:10.1007/978-3-642-24508-4}}.

\bibitem{DBLP:journals/eccc/Khaniki20}
Erfan Khaniki.
\newblock On proof complexity of resolution over polynomial calculus.
\newblock {\em Electron. Colloquium Comput. Complex.}, page~34, 2020.
\newblock URL: \url{https://eccc.weizmann.ac.il/report/2020/034}.

\bibitem{kraj_2019}
Jan Kraj{\'\i}{\v c}ek.
\newblock {\em Proof Complexity}.
\newblock Encyclopedia of Mathematics and its Applications. Cambridge
  University Press, 2019.
\newblock \href {https://doi.org/10.1017/9781108242066}
  {\path{doi:10.1017/9781108242066}}.

\bibitem{DBLP:journals/ipl/Lauria18a}
Massimo Lauria.
\newblock A note about \emph{k}-{DNF} resolution.
\newblock {\em Inf. Process. Lett.}, 137:33--39, 2018.
\newblock \href {https://doi.org/10.1016/j.ipl.2018.04.014}
  {\path{doi:10.1016/j.ipl.2018.04.014}}.

\bibitem{lidl_niederreiter_1996}
Rudolf Lidl and Harald Niederreiter.
\newblock {\em Finite Fields}.
\newblock Encyclopedia of Mathematics and its Applications. Cambridge
  University Press, 2 edition, 1996.
\newblock \href {https://doi.org/10.1017/CBO9780511525926}
  {\path{doi:10.1017/CBO9780511525926}}.

\bibitem{DBLP:journals/cc/PartT21}
Fedor Part and Iddo Tzameret.
\newblock Resolution with counting: Dag-like lower bounds and different moduli.
\newblock {\em Comput. Complex.}, 30(1):2, 2021.
\newblock \href {https://doi.org/10.1007/s00037-020-00202-x}
  {\path{doi:10.1007/s00037-020-00202-x}}.

\bibitem{MR897705}
A.~A. Razborov.
\newblock Lower bounds on the dimension of schemes of bounded depth in a
  complete basis containing the logical addition function.
\newblock {\em Mat. Zametki}, 41(4):598--607, 623, 1987.

\bibitem{DBLP:journals/tcs/SavickyZ00}
Petr Savick{\'{y}} and Stanislav {\v{Z}}{\'{a}}k.
\newblock A read-once lower bound and a (1, +k)-hierarchy for branching
  programs.
\newblock {\em Theor. Comput. Sci.}, 238(1-2):347--362, 2000.
\newblock \href {https://doi.org/10.1016/S0304-3975(98)00219-9}
  {\path{doi:10.1016/S0304-3975(98)00219-9}}.

\bibitem{DBLP:journals/eccc/ServedioV12}
Rocco~A. Servedio and Emanuele Viola.
\newblock On a special case of rigidity.
\newblock {\em Electron. Colloquium Comput. Complex.}, page 144, 2012.
\newblock URL: \url{https://eccc.weizmann.ac.il/report/2012/144}.

\bibitem{DBLP:journals/siamcomp/ShaltielV10}
Ronen Shaltiel and Emanuele Viola.
\newblock Hardness amplification proofs require majority.
\newblock {\em {SIAM} J. Comput.}, 39(7):3122--3154, 2010.
\newblock \href {https://doi.org/10.1137/080735096}
  {\path{doi:10.1137/080735096}}.

\bibitem{DBLP:conf/stoc/Smolensky87}
Roman Smolensky.
\newblock Algebraic methods in the theory of lower bounds for boolean circuit
  complexity.
\newblock In {\em Proceedings of the 19th Annual {ACM} Symposium on Theory of
  Computing, 1987, New York, New York, {USA}}, pages 77--82. {ACM}, 1987.
\newblock \href {https://doi.org/10.1145/28395.28404}
  {\path{doi:10.1145/28395.28404}}.

\bibitem{DBLP:journals/combinatorica/Yehudayoff11}
Amir Yehudayoff.
\newblock Affine extractors over prime fields.
\newblock {\em Comb.}, 31(2):245--256, 2011.
\newblock \href {https://doi.org/10.1007/s00493-011-2604-9}
  {\path{doi:10.1007/s00493-011-2604-9}}.

\end{thebibliography}

\end{document}